\documentclass[11pt]{article}

\usepackage{amssymb}
\usepackage{amsthm}

\usepackage{iopams_add}

\newcommand{\textcolor}[1]{}

\newtheorem{theorem}{Theorem}[section]   
\newtheorem{definition}[theorem]{Definition}
\newtheorem{proposition}[theorem]{Proposition}
\newtheorem{corollary}[theorem]{Corollary}

\newtheorem{example}{Example}
\newtheorem{remark}[theorem]{Remark}

\def\IC{{\bf C}}

\def\cR{{\cal R}}

\def\la{{\langle}}
\def\ra{{\rangle}}
 
\def\diag{{\rm diag}\,} 
\def\tr{{\rm tr}\,} 
\def\trn{{\rm tr}}

\newcommand{\bra}[1]{\mbox{$\left\langle #1 \right|$}}
\newcommand{\ket}[1]{\mbox{$\left| #1 \right\rangle$}}
\newcommand{\braket}[2]{\mbox{$\left\langle #1 | #2 \right\rangle$}}

\begin{document} 
\openup 1 \jot
\title
{Conditions for degradability of tripartite quantum states}

\author{Chi-Hang Fred Fung \\
  \multicolumn{1}{p{.7\textwidth}}{\centering\emph{%
  Department of Physics and Center of Theoretical and Computational Physics, University of Hong Kong, Pokfulam Road, Hong Kong}}
  \\
  \\
  Chi-Kwong Li \\
  \multicolumn{1}{p{.7\textwidth}}{\centering\emph{%
  Department of Mathematics, College of William \& Mary, Williamsburg, Virginia 23187-8795, USA}}
  \\
  \\
  Nung-Sing Sze \\
  \multicolumn{1}{p{.7\textwidth}}{\centering\emph{%
  Department of Applied Mathematics, The Hong Kong Polytechnic University, Hung Hom, Hong Kong}}
  \\
  \\
  H.~F. Chau \\
  \multicolumn{1}{p{.7\textwidth}}{\centering\emph{%
  Department of Physics and Center of Theoretical and Computational Physics, University of Hong Kong, Pokfulam Road, Hong Kong}}
}

\date{}

\maketitle

\begin{abstract}
Alice, Bob, and Eve share a pure quantum state.
{\textcolor{mycolor2}{
We introduce the notion of state degradability by asking
%We ask 
whether the joint density of Alice and Eve can be transformed to the joint density of Alice and Bob by processing Eve's part through a quantum channel, in order words, degrading Eve.
}}
We prove necessary and sufficient conditions for state degradability and provide an efficient method to quickly rule out degradability for a given state.
The problem of determining degradability of states is different from that of quantum channels, although the notion is similar.
{\textcolor{mycolor2}{
One application of state degradability is that it can be used to test channel degradability.
In particular, 
the degradability of the output state of a channel obtained from the maximally entangled input state gives information about the degradability of the channel.
}}
%In particular, 
%Also,
%we show that non-degradable (or non-anti-degradable) channels may generate degradable states.
\end{abstract}

\section{Introduction}

In quantum information processing, 
information is often encoded in quantum states which are transformed under quantum computation in order to carry out tasks such as the generation of secret keys~\cite{Bennett1984,Ekert1991}, encoding of error correcting codes~\cite{Shor:1995:QECC,Gottesman:1996:stabilizer,Calderbank:1997:stabilizer,Calderbank:1998:stabilizer}, and secret sharing~\cite{Hillery:1999:QSS,Karlsson:1999:QSS}.
The general quantum state transformation problem concerns whether a state can be transformed via a quantum process to another state, possibly with some constraint on the quantum process.
In as early as 1980's, 
Alberti and Uhlmann~\cite{Alberti1980163} studied the conditions for transforming two qubit mixed states.
Subsequently, conditions for the transformations between two sets of pure states without any restriction on the number of states were found~\cite{Uhlmann1985,Chefles200014,PhysRevA.65.052314,Chefles_Jozsa_Winter_2004}.
The transformation of entangled states under the condition that the two parties perform local operations has also been studied for a single bipartite state~\cite{Bennett1996,Lo:2001:concentration,PhysRevLett.83.436,PhysRevLett.83.1455,He2008,Gheorghiu2008} and for multiple bipartite states~\cite{Chau:2012:enttrans}.
Extension to transformations for more than two parties has also been considered~\cite{Gour:2011:multipartite}.

{\textcolor{mycolor2}{
In this paper, 
we introduce the notion of state degradability, which is based on
%We study 
a transformation problem where we ask whether a subsystem can be degraded to another subsystem.
}}
More precisely,
consider a quantum state in $|\psi\ra \in H_A \otimes H_B$ shared by Alice and Bob.
Assume that this state is processed by Eve and becomes an entangled state
$|\tilde \psi \ra \in H_A \otimes H_B \otimes H_E$.
(In the context of quantum key distribution (QKD), such processing corresponds to eavesdropping by Eve or the noisy effect of the channel.)
We are interested in
constructing a quantum process $T: H_A \otimes H_E \rightarrow H_A \otimes H_B$ of 
the form
\begin{equation}
\label{TX}
T(X) = \sum_{j=1}^r (I_A \otimes F_j)X(I_A\otimes F_j)^{*} \
\hbox{ with } \ F_j: H_E\rightarrow H_B \ \hbox{ satisfying } \  
\sum_{j=1}^r F_j^{*}F_j = I_E
\end{equation}
such that 
%{\textcolor{mycolor}{
$T(\rho_{AE}) = \rho_{AB}$ for
$$\rho_{AE} = \tr_B(\rho) \in {\mathcal B}(H_A \otimes H_E), \
\rho_{AB} = \tr_E(\rho) \in {\mathcal B}(H_A\otimes H_B) \ 
\hbox { with } \rho = |\tilde \psi\ra \la \tilde \psi |, $$
where ${\mathcal B}(H)$ is the set of bounded, positive-semidefinite operators acting on $H$, 
%}}
%{\textcolor{mycolor2}{
and $r$ is the number of Kraus operators of the quantum channel $T$ which can be arbitrary.
%}}
If such a map $T$ exists, we call the state $|\tilde \psi \ra$ $E \rightarrow B$ degradable.
Similarly, if there exists a quantum channel $T': H_A \otimes H_B \rightarrow H_A \otimes H_E$ %such that $T'(X_3) = X_2$, we call the state
%{\textcolor{mycolor}{
such that $T'(\rho_{AB}) = \rho_{AE}$, we call the state
%}}
$B \rightarrow E$ degradable.
A state may be $E \rightarrow B$ and $B \rightarrow E$ degradable.

It is interesting to know whether a state is degradable.
For example,
in QKD if the joint state between Alice and Bob can be shown to be the same as the joint state between Alice and Eve via some processing of Eve's part, then no secret key can be generated with one-way postprocessing~\cite{Bruss:1998:cloning,Nowakowski:2009:testforcapacity,Moroder:2006:upperbound}.
%{%\bf
%\textcolor{mycolor}
%{
Also, state degradability is related to asymmetric quantum cloning~\cite{Cerf:2000:cloning,Fiurasek:2005:cloning,Iblisdir:2005:cloning}, in which
the two output subsystems are not necessarily copies of each other, but one subsystem can be transformed to be a clone of the other.
If a given state is degradable, it means it could have been produced by asymmetric cloning of some other state.
%}}

Degradability has been studied in the context of quantum channels~\cite{Devetak:2005:degradable,Cubitt:2008:degradable}.
Let us consider a channel in system 
%{\textcolor{mycolor}{
$B$
%}}
which is described as a unitary transformation with ancillary system 
%{\textcolor{mycolor}{
$E$
%}}
prepared in a standard state:
{\textcolor{mycolor}{
$$
\Phi_B(\rho_B) = \tr_E [ U_{BE} ( \rho_B \otimes \ket{0}_E\bra{0} ) U_{BE}^* ] .
$$
Note that this system $A$ does not appear in this definition of degradable channel.
}}%
This induces the complementary channel
{\textcolor{mycolor}{
$$
\Phi_E(\rho_B) = \tr_B [ U_{BE} ( \rho_B \otimes \ket{0}_E\bra{0} ) U_{BE}^* ].
$$
}}%
The channel 
{\textcolor{mycolor}{$\Phi_B$}}
is called degradable when it may be degraded to 
{\textcolor{mycolor}{$\Phi_E$,}}
that is, there exists a quantum channel $\hat{T}: H_B \rightarrow H_E$ such that 
{\textcolor{mycolor}{$\hat{T} \circ \Phi_B = \Phi_E$.}}
Similarly, 
{\textcolor{mycolor}{$\Phi_B$}}
is called anti-degradable when there exists a quantum channel $\hat{T}: H_E \rightarrow H_B$ such that 
{\textcolor{mycolor}{$\hat{T} \circ \Phi_E = \Phi_B$.}}

It is clear that a degradable (anti-degradable) channel 
always output a state that is 
{\textcolor{mycolor}{$B \rightarrow E$ ($E \rightarrow B$)}}
degradable for any input.
On the other hand, there are channels that output a degradable state for some input and a non-degradable state for another input.
For example, consider this channel:
\begin{eqnarray*}
\ket{0}_B \ket{0}_E &\rightarrow& \ket{00}_{BE}
\\
\ket{1}_B \ket{0}_E &\rightarrow& \ket{11}_{BE}
\\
\ket{2}_B \ket{0}_E &\rightarrow& \ket{10}_{BE} .
\end{eqnarray*}
For the input $\ket{00}_{AB}+\ket{11}_{AB}$, we get the output
$\ket{000}_{ABE}+\ket{111}_{ABE}$
which is $B \rightarrow E$ and $E \rightarrow B$ degradable.
But for the input $\ket{00}_{AB}+\ket{11}_{AB}+\ket{22}_{AB}$, we get
$\ket{000}_{ABE}+\ket{111}_{ABE}+\ket{210}_{ABE}$.
For this state, since to degrade $E$ to $B$, $\ket{0}_E$ has to change to $\ket{0}_B$ and $\ket{1}_B$, which is not possible without knowing $A$.
Thus, the output state is not $E \rightarrow B$ degradable.
With a similar argument, it is also not $B \rightarrow E$ degradable.
This shows that a channel may output both degradable and non-degradable states.
Therefore, it is a new problem
to study degradable states without reference to whether the channel generating that state is degradable or not.

{\textcolor{mycolor2}{
As one application of state degradability, we prove in Sec.~\ref{sec-connection} that state degradability can be used to test channel degradability.
In particular, 
the degradability of the output state of a channel obtained from the maximally entangled input state gives information about the degradability of the channel.
We show that if the channel output state of the maximally entangled input state is degradable, then the corresponding output state of the channel is degradable for any input state.
Also, if the channel output state of the maximally entangled input state is not degradable, then the corresponding output state is not degradable for any input state in a special class.
}}

{%\bf
\textcolor{mycolor}{
In some applications, 
the issue of state degradability arises naturally in that
%the situation is that 
Alice, Bob, and Eve are initially given a tripartite state, without regard to the details of how it is given.
This may occur due to, for example, an entanglement source generating a tripartite state, or an unknown quantum channel processing one part of a bipartite input.
As a specific example, in entanglement distillation~\cite{Bennett:1996:EDP}, the problem is often cast as that given a noisy state in $AB$, which is purified to a tripartite state in $ABE$, the goal is to transform it (through, e.g., local operations and classical communications) to a maximally entangled state.
This can be viewed as a problem of a given initial state.
A similar situation occurs in QKD~\cite{Bennett1984,Ekert1991}.
After the quantum state transmission step, Alice and Bob are given bipartite states which they would like to transform to a secret key.
They first learn about their states by error testing and then choose the appropriate procedures to correct bit errors and amplify privacy.
This is also a problem centered on a given state.
And as mentioned before, if the state is degradable, no secret key can be generated with one-way postprocessing~\cite{Bruss:1998:cloning,Nowakowski:2009:testforcapacity,Moroder:2006:upperbound}; and thus no maximal entanglement can be distilled.
}}

We formulate the mathematical problem as follows.

\smallskip\noindent
{\bf State-Degradability Problem} Let  $x \in \IC^n \otimes \IC^p \otimes \IC^q$. 
Let $X_i = \tr_i(xx^*)$ with $i = 1,2,3$, 
be the partial traces of $xx^*$ in the three subsystems:
$\IC^p\otimes \IC^q$, $\IC^n \otimes \IC^q$, and $\IC^n\otimes \IC^p$.
Determine conditions on $x$ (or a class of $x$) such that there is $T$ of the form
(\ref{TX}) such that $T(X_2) = X_3$.
{\textcolor{mycolor}{
Here, we adopt the mathematical notation: $X_1=\rho_{BE}$, $X_2=\rho_{AE}$, and $X_3=\rho_{AB}$.
It turns out that this notation allows our mathematical results in the following sections to be concisely described.
We will however switch back to the physicist notation of $A,B,E$ when we discuss examples of physical relevance.
}}%
{\textcolor{mycolor2}{
Also, we use the notations $M_p$ to denote the set of $p \times p$ matrices and $M_{p,q}$ the set of $p \times q$ matrices.
In this paper, we consider systems of finite dimensions.
}}

The problem of determining whether a state is degradable is similar to the problem of finding a symmetric extension of a state~\cite{Nowakowski:2009:testforcapacity} in that both problems are characterized by the generation of Alice and Bob's state by processing
Eve's part of Alice and Eve's state.
However, in the latter problem, we are given Alice and Bob's state and we seek an extension that adds Eve to the overall state, while in our state degradability problem, a tripartite state is given initially.

We first give low dimension examples in Sec.~\ref{sec-low-dim-examples} which helps to understand the nature of the problem.
Then, in Sec.~\ref{sec-general-result}, we prove necessary and sufficient conditions for our state-degradability problem as stated above.
{\textcolor{mycolor}{
We discuss the physical interpretation of the transformability conditions in Sec.~\ref{sec-physical-interpretation}.
Sec.~\ref{sec-connection} 
discusses state degradability when a quantum channel is used to generate the overall state.
}}
The transformability conditions might be difficult to verify in general, and thus we propose an easily computable method that can quickly rule out degradability of a given state in Sec.~\ref{sec-rule-out}.
We discuss some additional problems and observations in Sec.~\ref{sec-additional-problems}.
Finally, we conclude in Sec.~\ref{sec-conclusion}.

\section{Low dimension examples}
\label{sec-low-dim-examples}
  
\begin{example} 
\label{example-low-dim-1}
\rm
Suppose $x = (x_{1},  \dots, x_8)^t \in \IC^8 \equiv \otimes^3(\IC^2)$.
Let  
%{\textcolor{mycolor}{
$xx^* \in M_8$ 
%}}
and $\tr_1, \tr_2, \tr_3$ be the partial trace on the three
systems. Then
$$X_1 = \tr_1(xx^*) = (x_1 x_2 x_3 x_4)^t(\bar x_1 \bar x_2 \bar x_3 \bar x_4) +
(x_5 x_6 x_7 x_8)^t(\bar x_5 \bar x_6 \bar x_7 \bar x_8),$$
$$X_2 = \tr_2(xx^*) = (x_1 x_2 x_5 x_6)^t(\bar x_1 \bar x_2 \bar x_5 \bar x_6) +
(x_3 x_4 x_7 x_8)^t(\bar x_3 \bar x_4 \bar x_7 \bar x_8),$$
$$X_3 = \tr_3(xx^*) = (x_1 x_3 x_5 x_7)^t(\bar x_1 \bar x_3 \bar x_5 \bar x_7) +
(x_2 x_4 x_6 x_8)^t(\bar x_2 \bar x_4 \bar x_6 \bar x_8),$$
%{\textcolor{mycolor2}{
where for ease of notations, we omitted the commas in the vectors such as $(x_1, x_2, x_3, x_4)$.
%}}
We would like to know whether there is $T: M_4 \rightarrow M_4$ of the form
$$T(X) = \sum_j (I_2 \otimes F_j)X(I_2 \otimes F_j)^*$$
with $\sum_j F_j^* F_j = I_2$ such that 
$T(X_2) = X_3.$
\end{example}

\medskip\noindent
\begin{example} 
\label{example-low-dim-2}
\rm Let $x=(a,0,b,0,0,a,0,-b)^t$ with $2(a^2 + b^2)= 1$.
This state is non-trivial since it is not symmetric in 2 and 3.
Then 
$$X_2 = 
\pmatrix{a^2+b^2 & 0 & 0 & a^2-b^2 \cr 0 & 0 & 0 & 0 \cr 
0 & 0 & 0 & 0 \cr a^2-b^2 &0 & 0 & a^2+b^2\cr}
\qquad \hbox{ and } \qquad 
X_3 =
\pmatrix{a^2 & ab & 0 & 0 \cr ab & b^2 & 0 & 0 \cr 
0 & 0 & a^2 & -ab \cr 0 & 0 & -ab & b^2\cr}.$$
Let 
$$F_1 = \pmatrix{a&a\cr b&-b\cr} \qquad \hbox{ and } \qquad 
F_2 = \pmatrix{a&-a\cr b& b\cr}.$$
Then $F_1^*F_1 + F_2^*F_2 = I_2$, and $T(X_2) = X_3$ if
$$T(X) = (I_2 \otimes F_1)X_2(I_2 \otimes F_1)^* + (I_2 \otimes F_2)X_2(I_2 \otimes F_2)^*.$$
\end{example}

\medskip\noindent
\begin{proposition} 
\label{proposition-low-dim-3}
{\textcolor{mycolor2}{
In Examples~\ref{example-low-dim-1} and \ref{example-low-dim-2},
}}%
the desired map 
exists if and only if 
there are $F_1, \dots, F_r$ with $\sum_{j=1}^r F_j^*F_j = I_2$ 
such that the map $L: M_2 \rightarrow M_2$ defined by
$$L(X) = \sum_{j=1}^r F_j X F_j^*$$ 
satisfies $L(R_iR_j^*) = S_iS_j^*$ for $1 \le i, j \le 2$, where
$$R_1 = \pmatrix{x_1 & x_3 \cr x_2 & x_4 \cr},  \ R_2 = \pmatrix{ x_5 & x_7 \cr x_6 & x_8 \cr}, \
S_1 = \pmatrix{x_1 & x_2 \cr x_3 & x_4 \cr},  \ S_2 = \pmatrix{ x_5 & x_6 \cr x_7 & x_8 \cr}.$$
\end{proposition}

\medskip\noindent
{\bf Remark} Alternatively, we can check whether there is a TPCP map $T$ sending 
$R_1R_1^*, R_2R_2^*, (R_1+R_2)(R_1+R_2)^*, (R_1+iR_2)(R_1+iR_2)^*$ to
$S_1S_1^*, S_2S_2^*, (S_1+S_2)(S_1+S_2)^*, (S_1+iS_2)(S_1+iS_2)^*$.
This can be checked readily.

\medskip\noindent
\begin{example}
\label{example-low-dim-3}
\rm
Suppose $x = (x_1, \dots, x_{16}) \in \IC^2\otimes \IC^2 \otimes \IC^4$.
Let  
{\textcolor{mycolor}{
$xx^* \in M_{16}$
}}%
and $\tr_1, \tr_2, \tr_3$ be the partial trace on the three
systems. Then
$$X_1 = \tr_1(xx^*) = (x_1 \cdots x_8)^t(\bar x_1 \cdots \bar x_8) +
(x_9 \cdots  x_{16})^t(\bar x_9 \cdots \bar x_{16}),$$ 
$$X_2 = \tr_2(xx^*) =  u_1u_1^* + u_2u_2^* = [u_1 u_2][u_1 u_2]^*,$$
with
$$u_1 = (x_1 x_2 x_3 x_4 x_9 x_{10} x_{11} x_{12})^t, 
\ u_2 = (x_5 x_6 x_7 x_8 x_{13} x_{14} x_{15} x_{16})^t,$$
and 
$$X_3 = \tr_3(xx^*) = v_1v_1^* + \cdots + v_4v_4^* = [v_1 \cdots v_4][v_1 \cdots v_4]^*$$
$$v_1 = (x_1 x_5 x_9 x_{13})^t, \ 
v_2 = (x_2 x_6 x_{10} x_{14})^t, \
v_3 = (x_3 x_7 x_{11} x_{15})^t, \
v_4 = (x_4 x_8 x_{12} x_{16})^t.$$
We would like to know whether there is $T: M_8 \rightarrow M_4$ of the form
$$T(X) = \sum_j (I_2 \otimes F_j)X(I_2 \otimes F_j)^*$$
with $\sum_j F_j^*F_j = I_4$ such that 
$T(X_2) = X_3.$
\end{example}

\begin{proposition} 
\label{proposition-low-dim-4}
%In Example~\ref{example-low-dim-3},
{\textcolor{mycolor2}{
Using the notation of Example~\ref{example-low-dim-3}, let
}}%
%Use the above notation. 
%Let 
$$R_1 = \pmatrix{x_1 & x_5 \cr x_2 & x_6 \cr x_3 & x_7 \cr x_4 & x_8 \cr}, \quad
R_2 = \pmatrix{x_9 & x_{13} \cr x_{10} & x_{14} \cr x_{11} & x_{15} \cr x_{12} & x_{16} \cr},$$
$$S_1 = \pmatrix{x_1 & x_2 & x_3 & x_4 \cr x_5 & x_6 & x_7 & x_8 \cr} = R_1^t, \quad
S_2 = \pmatrix{x_9 & x_{10}& x_{11} & x_{12} \cr x_{13} & x_{14} & x_{15} & x_{16} \cr} = R_2^t.$$
Then, 
{\textcolor{mycolor2}{
in Example~\ref{example-low-dim-3},
}}%
the desired map exists if and only if there exists a TPCP map sending 
$$R_1R_1^*, R_2R_2^*, (R_1+R_2)(R_1+R_2)^*, (R_1+iR_2)(R_1+iR_2)^*$$ to
$$S_1S_1^*, S_2S_2^*, (S_1+S_2)(S_1+S_2)^*, (S_1+iS_2)(S_1+iS_2)^*.$$
\end{proposition}

{\textcolor{mycolor}{
We remark that Propositions~\ref{proposition-low-dim-3} and \ref{proposition-low-dim-4} are special cases of Theorem~\ref{thm-main1} below.
}}

\section{General result}
\label{sec-general-result}

Suppose
$x = (x_{ijk})\in \IC^n \otimes \IC^p \otimes \IC^q$. We always assume that the entries of
$x$ are arranged in lexicographic (dictionary) order of the indexes $(ijk)$, i.e., 
$x_{111}$ is the first entry and $x_{npq}$ is the last entry.

\begin{theorem}
\label{thm-main1}
Suppose
$x = (x_{ijk})\in \IC^n \otimes \IC^p \otimes \IC^q$. 
Then 
$\tr_1(xx^*) =  \sum_{i=1}^n (x_{ijk})(x_{ijk})^* \in M_{p,q}$,
$$X_2 = \tr_2(xx^*) =  \sum_{j=1}^p (x_{ijk})(x_{ijk})^* \in M_{n,q}, \quad
X_3 = \tr_3(xx^*) =  \sum_{k=1}^q (x_{ijk})(x_{ijk})^* \in M_{n,q}.$$
Let 
$$S_i = (x_{ijk})_{1 \le j \le p, 1 \le k \le q} \in M_{p,q} 
\quad \hbox{ and } \quad R_i = S_i^t \in M_{q,p} \quad \hbox{ for } i = 1, \dots, n.$$
Suppose $R_i$ has rank $k_i$
for $i = 1, \dots, n$.
Set $R_i = U_iD_iV_i^t$ such that $D_i\in M_{k_i}$ is a diagonal matrix with  
positive diagonal entries arranged in descending order,
$U_i$ and $V_i$ have orthonormal columns.
Then the following conditions are equivalent.

{\rm (a)} There is a TPCP map $T: M_n\otimes M_q \rightarrow M_n\otimes M_p$
of the form 
\begin{equation} \label{form2}
X\mapsto \sum_{j=1}^r (I_n\otimes F_j)X(I_n\otimes F_j)^* \quad \hbox{ with } \quad 
F_1, \dots, F_r \in M_{p,q}
\end{equation}
satisfying $T(X_2) = X_3$.

{\rm (b)}
There is a TPCP map sending
$$\{R_uR_v^*: 1 \le u, v \le n\} \qquad \hbox{ to } \qquad
\{S_uS_v^*: 1 \le u, v \le n\}.$$

{\rm (c)} There are $p\times q$ matrices $F_1, \dots, F_r$ with $\sum_{j=1}^r F_j^*F_j = I_q$
and $k_i \times p$ matrices $W_{i1}, \dots, W_{ir}$ 
such that for all $i, j = 1, \dots, n$,
$$[F_1R_i \cdots F_r R_i] = V_iD_i[W_{i1} \cdots W_{ir}]
\quad \hbox{ and } \quad [W_{i1} \cdots W_{ir}][W_{j1} \cdots W_{jr}]^* = U_i^t\overline{U}_j.$$
\end{theorem}
 
\it Proof. \rm Direct checking shows that 
$X_2 = (R_uR_v^*)_{1 \le u,v \le n}$ and $X_3 = (S_u S_v^*)_{1 \le u,v\le n}$.
The map in the form (\ref{form2}) will send $X_2$ to $X_3$ if and only if
$\sum_{j=1}^r (F_jR_u R_v^*F_j) = (S_u S_v^*)_{1 \le u,v \le n}$.
Equivalently, the TPCP map $Y \mapsto \sum_{j=1}^r F_jYF_j^*$ will send
$R_uR_v^*$ to $S_uS_v^*$ for $1 \le u,v \le n$.
Thus, (a) and (b) are equivalent.

Suppose (b) holds. Then  for any $1 \le i, j \le n$,
$$\sum_\ell F_\ell R_i R_j^* F_\ell^* = V_i D_iU_i^t \overline{U}_jD_jV_j^*.$$
Considering $i = j$, we see that  ${\rm col}(F_\ell U_iD_i) \subseteq {\rm col}(V_iD_i)$, 
where col$(X)$ denotes the column space of $X$. Thus,
$F_\ell R_i = V_i D_i W_{i\ell}$ for a suitable $k_i\times p$ matrix $W_{i\ell}$.
Now, 
$$\sum_\ell  F_\ell R_iR_j^* F_\ell^* 
= V_iD_i[W_{i1} \cdots W_{ir}][W_{j1} \cdots W_{jr}]^* D_j V_j^* 
= V_iD_iU_i^t\overline{U}_jD_jV_j^*.$$
Multiplying $D_i^{-1}V_i^*$ to the left and multiplying 
$V_jD_j$ to the right, see that 
$$[W_{i1} \cdots W_{ir}][W_{j1} \cdots W_{jr}]^* = U_i^t\overline{U}_j$$ 
as asserted in (c).
The converse can be checked directly.
\qed

\begin{theorem}
\label{thm-main2}
Use the notation in Theorem~\ref{thm-main1}.
Suppose $R_i = u_id_iv_i^t$ is rank one for $i = 1, \dots, n$.
Then conditions {\rm (a) - (c)} in Theorem~\ref{thm-main1} are equivalent to the following.

{\rm (d)} There are unit vectors $\gamma_1, \dots, \gamma_n \in \IC^r$ and a unitary $U$ such that
$$U[e_1 \otimes u_1 \dots e_1 \otimes u_n] = [\gamma_1 \otimes v_1 \cdots \gamma_n\otimes v_n].$$

{\rm (e)} There is a correlation matrix $C$ such that $(u_i^*u_j) = (v_i^*v_j) \circ C$.

{\rm (f)} There exists a TPCP map sending $u_iu_i^*$ to $v_iv_i^*$ for $i = 1, \dots, n$.

{\textcolor{mycolor2}{
\noindent Here, we abuse the notation of $e_1 \otimes u_i$ to represent a vector in $\IC^{pr}$ with the first $q$ elements being $u_i$ and the remaining elements being zero.
}}%
\medskip\noindent
\end{theorem}

\it Proof. \rm Using condition (c) and focusing on  the column space and row space 
of $F_\ell u_jd_jv_j^t = v_j w_{j\ell}^t$, we see that 
$w_{j\ell} = d_j\gamma_{j\ell} v_j$ 
for some $\gamma_{j\ell} \in \IC$ for $j = 1, \dots, n$,
and 
$$(u_i^t \overline{u}_j) = ((\gamma_i \otimes v_i)^t (\overline{\gamma}_j \otimes\overline{v}_j))
= (\gamma_i^t \overline{\gamma_j})\circ (v_i^t \overline{v}_j),$$ 
where $\gamma_i = (\gamma_{i1}, \dots, \gamma_{ir})^t$
is a unit vector for $i = 1, \dots, n$.
Thus, there is a unitary $U \in M_{pr}$ such that 
$U[e_1 \otimes u_1 \dots e_1 \otimes u_n] = [\gamma_1 \otimes v_1 \cdots \gamma_n \otimes v_n]$.
If (d) holds, one can check condition (c) readily.

The equivalence of (d), (e), (f) follow from the results in Ref.~\cite{Huang:2012}.
\qed 

\begin{corollary} 
\label{cor-main3}
Use the notation in Theorem~\ref{thm-main2}.
The following are equivalent.

{\rm (a)} There are 
TPCP maps $T = I_n\otimes T_1$ and $L=I_n \otimes L_1$  
such that  $T(X_2) = X_3$ and $L(X_3) = X_2$.

{\rm (b)} There is a diagonal unitary matrix $E \in M_n$
such that   $(u_i^*u_j) = E^*(v_i^*v_j)E$.

{\rm (c)} We may enlarge $[u_1 \cdots u_n]$ and $[v_1 \dots  v_n]$ by adding zero rows to
get $m\times n$ matrices $\tilde U$ and $\tilde V$  
with $m = \max\{p,q\}$ 
such that $W\tilde U = \tilde V E$ for 
a unitary $W \in M_m$ and a diagonal unitary $E \in M_n$. 
\end{corollary}

In 
Example~\ref{example-low-dim-2},
we have $R_1 = aE_{11}+bE_{22}$ and $R_2 = aE_{21}-bE_{22}$.
There is TPCP map sending $X_3$ to $X_2$ if and only if
$a^2 - b^2 = 0$. In such a case, we can set
$L(X) = (I_2\otimes G)X(I_2\otimes G)^*$
with $G  = (a^2+b^2)^{-1/2}\pmatrix{a & b \cr a & -b\cr} = \sqrt{2} \pmatrix{a & b \cr a & -b\cr}$.
By the fact that $a^2+b^2 = 1/2$ and $a^2 - b^2 = 0$, we see that
$L(X_3) = X_2$.

\section{Physical interpretation 
of the transformability conditions
}
\label{sec-physical-interpretation}

Using the physics notation, Alice, Bob, and Eve share a tripartite pure state $\ket{\Psi}_{ABE}$ which corresponds to $x$ in Sec.~\ref{sec-general-result} with $x_{ijk}=\braket{ijk}{\Psi}_{ABE}$ where $\ket{ijk}_{ABE}$ is an eigenstate in the computational basis (note that the indexes start at $1$ instead of the usual $0$).
Define $\ket{\Psi_i}_{BE} \triangleq (\bra{i}_A \otimes I_{BE})\ket{\Psi}_{ABE}$, which is a state conditional on $A$ being $\ket{i}_A$.
Thus, we have
$$
\ket{\Psi}_{ABE}=\sum_{i=1}^n \ket{i}_A \ket{\Psi_i}_{BE}.
$$
According to the definitions of $R_i$ and $S_i$,
\begin{eqnarray*}
R_i R_i^*=\tr_{B} ( \ket{\Psi_i}_{BE} \bra{\Psi_i} ) \triangleq \rho_E^{(i)} && \hbox{and }\\
S_i S_i^*=\tr_{E} ( \ket{\Psi_i}_{BE} \bra{\Psi_i} ) \triangleq \rho_B^{(i)} && \hbox{for } i=1,\dots,n.
\end{eqnarray*}
In other words, $\rho_E^{(i)}$ is the reduced density matrix of $E$ conditioned on $A$ being $\ket{i}_A$; similarly for $\rho_B^{(i)}$.
{\textcolor{mycolor}{
Note that for the rest of this section, we use the physics notation of $A,B,E$ to label states.
}}

If Eve can imitate Bob using quantum channel $\mathcal E$ (i.e., ${\mathcal E}(\rho_{AE})=\rho_{AB}$), then
\begin{eqnarray}
\label{eqn-transform-any-projection}
\bra{\phi}_A {\mathcal E}(\rho_{AE}) \ket{\phi}_A= 
\bra{\phi}_A \rho_{AB} \ket{\phi}_A 
\end{eqnarray}
for any $\ket{\phi}_A$.
Thus, when the projections are on the computational basis for $A$, the
quantum channel is able to transform $\rho_E^{(i)}$ to $\rho_B^{(i)}$,
i.e.,
\begin{equation}
\label{eqn-transform-computation-basis}
{\mathcal E} (R_i R_i^*)=S_i S_i^* ,
\end{equation}
for $i=1,\dots,n$.
On the other hand,
the transformability condition of Theorem~\ref{thm-main1} (b) includes additional cross terms (i.e., $R_u R_v^* \rightarrow S_u S_v^* $ for $u \ne v$).
Essentially, the transformability of the cross terms guarantees the transformability of $E$ to $B$ in other bases.
To see this, consider the $\{+,-\}$ complementary basis for $A$ where we define $\ket{\pm}_A=(\ket{1}_A \pm \ket{2}_A)/\sqrt{2}$.
If Eve is able to pretend to be Bob, 
\eqref{eqn-transform-any-projection} means that 
the transformation
$\rho_E^{(\pm)} \rightarrow \rho_B^{(\pm)}$ is possible,
where
$\rho_E^{(\pm)} \triangleq \tr_{B} ( \ket{\Psi_\pm}_{BE} \bra{\Psi_\pm} ) $,
$\rho_B^{(\pm)} \triangleq \tr_{E} ( \ket{\Psi_\pm}_{BE} \bra{\Psi_\pm} ) $, and
$\ket{\Psi_\pm}_{BE}=(\ket{\Psi_1}_{BE} \pm \ket{\Psi_2}_{BE})/\sqrt{2}$.
In other words,
\eqref{eqn-transform-any-projection} becomes
\begin{eqnarray}
\label{eqn-basis+--transform0}
&&{\mathcal E}(\rho_E^{(\pm)}) = \rho_B^{(\pm)}
\\
&\Leftrightarrow&
\nonumber
\\
&&{\mathcal E}(R_1 R_1^* + R_2 R_2^* \pm R_1 R_2^* \pm R_2 R_1^* ) = S_1 S_1^* + S_2 S_2^* \pm S_1 S_2^* \pm S_2 S_1^* .
\label{eqn-transform-+-basis}
\end{eqnarray}
Here, the reduced density matrix of $E$ conditioned on $A$ being $\ket{\pm}_A$ is $\rho_E^{(\pm)}=R_\pm R_\pm^*$.
According to the definition of $R_i$ which is a rearragement of the elements of $\ket{\Psi_i}_{BE}$ for $i=+,-,1,2$,
$R_\pm=(R_1 \pm R_2)/\sqrt{2}$.
We have similar expressions for $B$.
Given \eqref{eqn-transform-computation-basis}, \eqref{eqn-transform-+-basis} is true if and only if 
\begin{eqnarray*}
{\mathcal E}( R_1 R_2^* + R_2 R_1^* ) = S_1 S_2^* + S_2 S_1^*.
\end{eqnarray*}
This shows that the transformability of the cross terms $R_u R_v^* \rightarrow S_u S_v^*$ for $u \neq v$ guarantees the transformability of $E$ to $B$ in other bases.
Therefore, one cannot simplify the condition checking of Theorem~\ref{thm-main1} (b) by ignoring the cross terms.
The following example illustrates this point by showing that there exists a state for which
$R_u R_v^* \rightarrow S_u S_v^*$ for $u=v$ but not $u \neq v$.

\begin{example}
{\rm

The initial state is a $3\times 2  \times 2$ system in $A$, $B$, and $E$:
\begin{eqnarray}
\ket{\Psi}_{ABE}&=&
\label{eqn-original-state1}
\ket{1}_A \otimes 
\left[
\begin{pmatrixQ}\alpha\\ \beta\end{pmatrixQ}_B
\otimes
\begin{pmatrixQ}a\\ b\end{pmatrixQ}_E
+
\begin{pmatrixQ}\alpha\\ -\beta\end{pmatrixQ}_B
\otimes
\begin{pmatrixQ}a\\ -b\end{pmatrixQ}_E
\right]
+
\nonumber
\\
&&
\ket{2}_A \otimes 
\begin{pmatrixQ}\alpha\\ \imath \beta\end{pmatrixQ}_B
\otimes
\begin{pmatrixQ}a\\ b\end{pmatrixQ}_E
+
\ket{3}_A \otimes 
\begin{pmatrixQ}\alpha\\ -\imath \beta\end{pmatrixQ}_B
\otimes
\begin{pmatrixQ}a\\ -b\end{pmatrixQ}_E
\\
&\triangleq&
\ket{1}_A \otimes
\Big[
\ket{p_+}_B \otimes \ket{\phi_+}_E
+
\ket{p_-}_B \otimes \ket{\phi_-}_E
\Big]
+
\nonumber
\\
&&
\ket{2}_A \otimes \ket{q_+}_B \otimes \ket{\phi_+}_E
+
\ket{3}_A \otimes \ket{q_-}_B \otimes \ket{\phi_-}_E
\end{eqnarray}
where $\alpha, \beta=\sqrt{1-\alpha^2}, a, b=\sqrt{1-a^2} \in \cR$, and $\imath=\sqrt{-1}$.

Our goal is to show that there exists a quantum channel $\mathcal E$ such that (i)
${\mathcal E}(\rho_E^{(j)})=\rho_B^{(j)}, j=1,2,3$, i.e.,
\begin{equation}
\label{eqn-R-diagonal-terms}
{\mathcal E}(R_j R_j^*)=S_j S_j^*, \:\: j=1,2,3,
\end{equation}
and (ii) there does not exist a quantum channel
$\mathcal E$ such that
\begin{equation}
\label{eqn-R-all-terms}
{\mathcal E}(R_j R_k^*)=S_j S_k^*, \:\: j,k=1,2,3.
\end{equation}
This means that \eqref{eqn-R-diagonal-terms} does not imply \eqref{eqn-R-all-terms}.

We show that \eqref{eqn-R-diagonal-terms} holds but \eqref{eqn-basis+--transform0} does not hold for the state in \eqref{eqn-original-state1}.

\subsubsection*{Proof of the validity of \eqref{eqn-R-diagonal-terms}}

First, we show that \eqref{eqn-R-diagonal-terms} holds.
Assume that 
$
\braket{\phi_+}{\phi_-} < \braket{q_+}{q_-}
$
and so there exists 
a quantum channel $\mathcal E$ that transforms
$\ket{\phi_\pm} \longrightarrow \ket{q_\pm}$.
This can be verified by comparing the Gram matrices of the initial set of states and the final one~\cite{Uhlmann1985,Chefles200014,PhysRevA.65.052314,Chefles_Jozsa_Winter_2004}.

The quantum channel $\mathcal E$ is equivalent to
a unitary transformation $U_{E E'}$ using an extended Hilbert space $E'$:
\begin{eqnarray}
\ket{\Psi'}_{ABEE'} 
&=&
U_{E E'} \ket{\Psi}_{ABE} \ket{0}_{E'}
\\
&=&
\ket{1}_A \otimes 
\ket{\Psi_1'}_{BEE'}
+
\ket{2}_A \otimes 
\ket{\Psi_2'}_{BEE'}
+
\ket{3}_A \otimes 
\ket{\Psi_3'}_{BEE'}
\end{eqnarray}
where $\ket{\Psi_j'}_{BEE'}= U_{E E'} \ket{\Psi_j}_{BE} \ket{0}_{E'}, j=1,2,3$.

Note that
${\mathcal E} (R_j R_j^*)=\tr_{B E'} ( \ket{\Psi_j'}_{BEE'} \bra{\Psi_j'} ) = {\mathcal E} (\rho_E^{(j)})$.

In order that ${\mathcal E} (\rho_E^{(2)})=\rho_B^{(2)}=\ket{q_+}\bra{q_+}$, 
$U_{E E'}$ must transform as
\begin{eqnarray*}
&&
U_{E E'} \ket{\Psi_2}_{BE} \ket{0}_{E'}
\\
&=&
U_{E E'} \ket{q_+}_B \ket{\phi_+}_E \ket{0}_{E'}
\\
&=&
\ket{q_+}_B \ket{q_+}_E \ket{x_+}_{E'},
\end{eqnarray*}
where $\ket{x_+}_{E'}$ is some normalized vector.
Similarly, ${\mathcal E} (\rho_E^{(3)})=\rho_B^{(3)}$ implies that
$$
U_{E E'} \ket{\Psi_3}_{BE} \ket{0}_{E'}
=
\ket{q_-}_B \ket{q_-}_E \ket{x_-}_{E'},
$$
where $\ket{x_-}_{E'}$ is some normalized vector.
Then, we have
\begin{eqnarray}
\label{eqn-state-after-channel-Z}
\ket{\Psi'}_{ABEE'} 
&=&
\ket{1}_A 
\Big[
\ket{p_+}_B \ket{q_+}_E \ket{x_+}_{E'}
+
\ket{p_-}_B \ket{q_-}_E \ket{x_-}_{E'}
\Big]
+
\nonumber
\\
&&
\ket{2}_A \ket{q_+}_B \ket{q_+}_E \ket{x_+}_{E'}
+
\ket{3}_A \ket{q_-}_B \ket{q_-}_E \ket{x_-}_{E'}
\end{eqnarray}
We now verify that ${\mathcal E} (\rho_E^{(1)})=\rho_B^{(1)}$.
The LHS is
\begin{eqnarray*}
{\mathcal E} (\rho_E^{(1)})
&=&
\tr_{B E'} ( \ket{\Psi_1'}_{BEE'} \bra{\Psi_1'} )
\\
&=&
\tr_{B E'} \Big[ P\big(
\ket{p_+}_B \ket{q_+}_E \ket{x_+}_{E'}
+
\ket{p_-}_B \ket{q_-}_E \ket{x_-}_{E'}
\big)
\Big]
\\
&=&
\ket{q_+}_E \bra{q_+}
+
\ket{q_-}_E \bra{q_-}
+
C
\ket{q_+}_E \bra{q_-}
+
C^*
\ket{q_-}_E \bra{q_+}
\end{eqnarray*}
where $P(\ket{\varphi})\triangleq\ket{\varphi}\bra{\varphi}$, and
$C \triangleq \braket{p_-}{p_+}_B
\braket{x_-}{x_+}_{E'}$.
Substituting the various vectors using \eqref{eqn-original-state1},
\begin{eqnarray}
{\mathcal E} (\rho_E^{(0)})
&=&
\begin{pmatrixQ}
2 \alpha ^2 &0\\0&2 \beta^2
\end{pmatrixQ}
+
C 
\begin{pmatrixQ}
\alpha ^2 & \imath \alpha \beta \\ \imath \alpha \beta & -\beta^2
\end{pmatrixQ}
+
C ^*
\begin{pmatrixQ}
\alpha ^2 & -\imath \alpha \beta \\-\imath \alpha \beta & -\beta^2
\end{pmatrixQ}.
\end{eqnarray}
The RHS is
\begin{eqnarray*}
\rho_B^{(1)}
&=&
\tr_{E} ( \ket{\Psi_1}_{BE} \bra{\Psi_1} )
\\
&=&
\tr_{E}
\Big[
P\big(
\ket{p_+}_B \ket{\phi_+}_E
+
\ket{p_-}_B \ket{\phi_-}_E
\big)
\Big]
\\
&=&
\ket{p_+}_B \bra{p_+}
+
\ket{p_-}_B \bra{p_-}
+
\braket{\phi_-}{\phi_+}_E
\ket{p_+}_B \bra{p_-}
+
\braket{\phi_+}{\phi_-}_E
\ket{p_-}_B \bra{p_+}
\\
&=&
\begin{pmatrixQ}
2 \alpha ^2 &0\\0&2 \beta^2
\end{pmatrixQ}
+
(a^2-b^2)
\begin{pmatrixQ}
2 \alpha ^2 &0\\0&-2 \beta^2
\end{pmatrixQ}.
\end{eqnarray*}
This means that
${\mathcal E} (\rho_E^{(1)})=\rho_B^{(1)}$ if and only if
$C=a^2-b^2$.
This is possible since we have assumed that
$
\braket{\phi_+}{\phi_-} < \braket{q_+}{q_-}=\braket{p_+}{p_-}
$.
We impose that $\ket{x_\pm}$ be chosen such that $C=a^2-b^2$, and thus
${\mathcal E} (\rho_E^{(j)})=\rho_B^{(j)}$ for $j=1,2,3$.

\subsubsection*{Invalidity of \eqref{eqn-basis+--transform0} for the state in \eqref{eqn-original-state1}}

We now show that \eqref{eqn-basis+--transform0} does not hold given that 
\begin{eqnarray}
\label{eqn-assumption1}
C=(\alpha^2-\beta^2)\braket{x_-}{x_+}=a^2-b^2.
\end{eqnarray}
Expressing \eqref{eqn-state-after-channel-Z} in the $\{+,-\}$ basis, we have
\begin{eqnarray*}
\ket{\Psi'}_{ABEE'} 
&=&
U_{EE'}\ket{\Psi}_{ABE}\ket{0}_{E'}
\\
&=&
\frac{1}{\sqrt{2}}
\ket{+}_A 
\Big[
(\ket{p_+}_B+\ket{q_+}_B) \ket{q_+}_E \ket{x_+}_{E'}
+
\ket{p_-}_B \ket{q_-}_E \ket{x_-}_{E'}
\Big]
+
\\
&&
\frac{1}{\sqrt{2}}
\ket{-}_A 
\Big[
(\ket{p_+}_B-\ket{q_+}_B) \ket{q_+}_E \ket{x_+}_{E'}
+
\ket{p_-}_B \ket{q_-}_E \ket{x_-}_{E'}
\Big]
+
\nonumber
\\
&&
\ket{3}_A \ket{q_-}_B \ket{q_-}_E \ket{x_-}_{E'}
\\
&\triangleq&
\ket{+}_A \otimes 
\ket{\Psi_+'}_{BEE'}
+
\ket{-}_A \otimes 
\ket{\Psi_-'}_{BEE'}
+
\ket{3}_A \otimes 
\ket{\Psi_3'}_{BEE'} .
\end{eqnarray*}
We show that ${\mathcal E} (\rho_E^{(-)})\neq\rho_B^{(-)}$.
The LHS is
\begin{eqnarray*}
2{\mathcal E} (\rho_E^{(-)})
&=&
2\tr_{B E'} ( \ket{\Psi_-'}_{BEE'} \bra{\Psi_-'} )
\\
&=&
\tr_{B E'} \Big[ P\big(
\ket{p_+'}_B \ket{q_+}_E \ket{x_+}_{E'}
+
\ket{p_-}_B \ket{q_-}_E \ket{x_-}_{E'}
\big)
\Big]
\\
&=&
\braket{p_+'}{p_+'}
\ket{q_+}_E \bra{q_+}
+
\ket{q_-}_E \bra{q_-}
+
D
\ket{q_+}_E \bra{q_-}
+
D^*
\ket{q_-}_E \bra{q_+}
\end{eqnarray*}
where $\ket{p_+'}_B=\ket{p_+}_B-\ket{q_+}_B$, and
$D \triangleq \braket{p_-}{p_+'}_B
\braket{x_-}{x_+}_{E'}$.
Substituting the various vectors using \eqref{eqn-original-state1},
\begin{eqnarray*}
\ket{p_+'}_B&=&\begin{pmatrixQ}0\\ (1-\imath)\beta\end{pmatrixQ}_B
\:\:\rm{and}
\\
2{\mathcal E} (\rho_E^{(-)})
&=&
2\beta^2
\begin{pmatrixQ}
\alpha^2 & -\imath \alpha \beta \\ \imath \alpha \beta & \beta^2
\end{pmatrixQ}
+
\begin{pmatrixQ}
\alpha^2 & \imath \alpha \beta \\-\imath \alpha \beta & \beta^2
\end{pmatrixQ}
+
D
\begin{pmatrixQ}
\alpha ^2 & \imath \alpha \beta \\ \imath \alpha \beta & -\beta^2
\end{pmatrixQ}
+
D^*
\begin{pmatrixQ}
\alpha ^2 & -\imath \alpha \beta \\-\imath \alpha \beta & -\beta^2
\end{pmatrixQ}
\end{eqnarray*}
where $D=-(1-{\imath})\beta^2 \braket{x_-}{x_+}$.
The RHS is
\begin{eqnarray*}
2\rho_B^{(-)}
&=&
2\tr_{E} ( \ket{\Psi_-}_{BE} \bra{\Psi_-} )
\\
&=&
\tr_{E}
\Big[
P\big(
\ket{p_+'}_B \ket{\phi_+}_E
+
\ket{p_-}_B \ket{\phi_-}_E
\big)
\Big]
\\
&=&
\ket{p_+'}_B \bra{p_+'}
+
\ket{p_-}_B \bra{p_-}
+
\braket{\phi_-}{\phi_+}_E
\ket{p_+'}_B \bra{p_-}
+
\braket{\phi_+}{\phi_-}_E
\ket{p_-}_B \bra{p_+'}
\\
&=&
\begin{pmatrixQ}
0 &0\\0&2 \beta^2
\end{pmatrixQ}
+
\begin{pmatrixQ}
\alpha ^2 &-\alpha \beta\\-\alpha \beta& \beta^2
\end{pmatrixQ}
+
(a^2-b^2)
\begin{pmatrixQ}
0 & (1+\imath) \alpha \beta \\ (1-\imath) \alpha \beta &-2 \beta^2
\end{pmatrixQ}.
\end{eqnarray*}
To show that
${\mathcal E} (\rho_E^{(-)})\neq\rho_B^{(-)}$, we compare their $(1,1)$ elements.
The RHS is $\alpha^2$, and the LHS is
$\alpha^2 (2\beta^2 +1 +D + D^*) = \alpha^2 (2\beta^2 +1 -2\beta^2 \braket{x_-}{x_+}_{E'})$.
However, due to the assumption in \eqref{eqn-assumption1}, $\braket{x_-}{x_+}_{E'}\neq 1$ in general.
Therefore, ${\mathcal E} (\rho_E^{(-)})\neq\rho_B^{(-)}$.

}
\end{example}

\section{
%{\textcolor{mycolor}{
Degradability for a given quantum channel
%}}
}
\label{sec-connection}

{\textcolor{mycolor}{
We discuss state degradability when the overall state is generated by a given quantum channel.
We show that if the channel output state of the maximally entangled input state is degradable, then the corresponding output state is degradable for any input state.
Also, if the channel output state of the maximally entangled input state is not degradable,
%the output state is not degradable for the maximally entangled input state, 
then the corresponding output state is not degradable for any input state in a special class.
\begin{theorem} 
\label{thm-connection-1}
Suppose that a state $|{\tilde{\psi}}\rangle_{ABE}$ is generated by processing a state $\ket{\psi}_{AB}$
by a channel 
$\Phi_B$
acting on subsystem $B$ with an ancilla in subsystem $E$.
The channel is implemented by a unitary extension $U_{BE}$ as follows:
$$
|{\tilde{\psi}}\rangle_{ABE}=
(I_A \otimes U_{BE}) \ket{\psi}_{AB} \ket{0}_E 
\triangleq U_{ABE} \ket{\psi}_{AB} \ket{0}_E ,
$$
where 
$\Phi_B (\tr_A (P(\ket{\psi}_{AB}))) = \tr_{AE}(P(|{\tilde{\psi}}\rangle_{ABE}))$
with
$P(\ket{\cdot})=\ket{\cdot}\bra{\cdot}$.
We assume that the dimensions of subsystems $A$ and $B$ are the same, $n$, so that the
maximally entangled state $\ket{\psi_M}_{AB}=\sum_{i=1}^n \ket{ii}_{AB}$ is defined.
Using $\ket{\psi_M}_{AB}$ as the input state,
if the output state 
$U_{ABE} \ket{\psi_M}_{AB} \ket{0}_E$
is $E \rightarrow B$ degradable with $T$ of the form (\ref{TX}) [i.e., $T$ satisfies
$T(\rho_{AE}) = \rho_{AB}$
where 
$\rho_{AE} = \tr_B(\rho) \hbox{ and } \rho_{AB} = \tr_E(\rho)   
\hbox { with } \rho = 
P(U_{ABE} \ket{\psi_M}_{AB} \ket{0}_E)$],
then the output state 
$U_{ABE} \ket{\psi}_{AB} \ket{0}_E$
is $E \rightarrow B$ degradable with the same $T$ for any input state $\ket{\psi}_{AB}$.
\end{theorem}
\begin{proof}
First, note that any state $\ket{\psi}_{AB}$ can be expressed as $\ket{\psi}_{AB} = (K_A \otimes I_B) \ket{\psi_M}_{AB}$ where $K_A=\sum_{i,j=1}^n \ket{j}_A \bra{i} \: \braket{ji}{\psi}_{AB}$.
Note that $K_A$ is not necessarily invertible.
 %and
Next, the condition for $E \rightarrow B$ degradability of the maximally entangled state means that
\begin{eqnarray}
\label{eqn-channel-degradability-1}
&&T(\rho_{AE}) = \rho_{AB}
\\
&\Rightarrow&
T( (K_A\otimes I) \rho_{AE} (K_A^* \otimes I))= (K_A \otimes I) \rho_{AB} (K_A^* \otimes I)
\hbox{ for any }K_A
\end{eqnarray}
where the last line is because $T$ acts only on subsystem $E$.
Finally, the term on the LHS is
\begin{eqnarray}
\label{eqn-channel-degradability-2}
(K_A\otimes I) \rho_{AE} (K_A^* \otimes I)
&=&
\tr_B\left[
P(U_{ABE} (K_A \otimes I_{BE}) \ket{\psi_M}_{AB} \ket{0}_E)
\right]
\\
&=&
\tr_B\left[
P(U_{ABE} \ket{\psi}_{AB} \ket{0}_E)
\right]
\end{eqnarray}
and we have an analogous term on the RHS.
Thus,
we have
\begin{eqnarray}
T(
\tr_B\left[
P(U_{ABE} \ket{\psi}_{AB} \ket{0}_E)
\right]
)
=
\tr_E\left[
P(U_{ABE} \ket{\psi}_{AB} \ket{0}_E)
\right]
\label{eqn-channel-degradability-3}
\end{eqnarray}
which means that the output state is $E \rightarrow B$ degradable for any input state $\ket{\psi}_{AB}$.
\end{proof}
\begin{corollary}
\label{corollary-connection-2}
For the channel
$\Phi_B$,
if the output state 
$U_{ABE} \ket{\psi_M}_{AB} \ket{0}_E$
is $E \rightarrow B$ degradable for the maximally entangled input state $\ket{\psi_M}_{AB}$, then the channel 
$\Phi_B$
is anti-degradable with respect to the complementary channel 
$\Phi_E$
[i.e.,
there exists a quantum channel $\hat{T}: H_E \rightarrow H_B$ such that 
$\hat{T} \circ \Phi_E = \Phi_B$].
Here, in terms of $U_{BE}$,
\begin{eqnarray}
\Phi_B(\rho_B)&=&\tr_E\left[
U_{BE} (\rho_B \otimes \ket{0}_E\bra{0}) U_{BE}^*
\right], \hbox{ and }
\\
\Phi_E(\rho_B)&=&\tr_B\left[
U_{BE} (\rho_B \otimes \ket{0}_E\bra{0}) U_{BE}^*
\right].
\end{eqnarray}
\end{corollary}
\begin{proof}
Since any state $\rho_B$ of dimension $n$ can be purified with a subsystem $A$ of dimension $n$ such that $\rho_B=\tr_A (P(\ket{\psi}_{AB}))$ for some $\ket{\psi}_{AB}$,
we can trace out subsystem $A$ on both sides of Eq.~\eqref{eqn-channel-degradability-3} to get
$\Phi_E$ processed by a channel acting on $E$ on the LHS and $\Phi_B$ on the RHS.
\end{proof}
\begin{remark}
In the proof of Theorem~\ref{thm-connection-1},
the operator $K_A$ performed on subsystem $A$ can have an operational interpretation related to entanglement transformation.
It may be viewed as a local filtering operation that can be implemented probabilistically by a quantum measurement.
This operation locally transforms a maximally entangled state to any other given state (entangled or unentangled) probabilistically.
\end{remark}
\begin{theorem}
\label{thm-connection-3}
Following the notations in Theorem~\ref{thm-connection-1}, 
for the input state $\ket{\psi_M}_{AB}$,
if the output state 
$U_{ABE} \ket{\psi_M}_{AB} \ket{0}_E$
is not $E \rightarrow B$ degradable,
then the output state 
$U_{ABE} \ket{\psi}_{AB} \ket{0}_E$
is not $E \rightarrow B$ degradable for any input state $\ket{\psi}_{AB}=(W_A \otimes I_B)\ket{\psi_M}_{AB}$ where $W_A$ is invertible.
\end{theorem}
\begin{proof}
We prove by contradiction.
Suppose that for some input state $\ket{\psi}_{AB}$,
the output state 
$U_{ABE} \ket{\psi}_{AB} \ket{0}_E$
is $E \rightarrow B$ degradable.
We repeat the arguments in the proof of Theorem~\ref{thm-connection-1} with $\ket{\psi}_{AB}$ and $\ket{\psi_M}_{AB}$ swapped and with $K_A=W_A^{-1}$.
Then, Eqs.~\eqref{eqn-channel-degradability-1}-\eqref{eqn-channel-degradability-3} follow,  concluding that when the input state is $\ket{\psi_M}_{AB}$, the output state is $E \rightarrow B$ degradable.
This contradicts the assumption and thus proves the theorem.
\end{proof}
}} % color

\section{Necessary condition for degradability}
\label{sec-rule-out}

We provide an easily computable method to rule out the degradability of a given state.
It is based on  
the expression of degradability in 
condition (b) of Theorem~\ref{thm-main1} and
the contractivity of quantum channels under the trace distance.

\begin{definition}
The trace norm of a matrix $\sigma \in M_q$ is $\tr |\sigma|=\sum_{j=1}^q \lambda_j$ where $\lambda_j$ are the singular values of $\sigma$.
\end{definition}

\begin{definition}
The trace distance between two matrices $\rho, \sigma \in M_q$ is
$d(\rho,\sigma)=\frac{1}{2} \tr |\rho-\sigma|$.
\end{definition}

Quantum channels are contractive under the trace distance for quantum states, i.e., 
$d(\rho,\sigma) \geq d({\mathcal F} (\rho),{\mathcal F} (\sigma) )$
for any density matrices $\rho$ and $\sigma$ and quantum channel ${\mathcal F}$~\cite{Ruskai:1994:contraction} (see also Theorem~9.2 of \cite{Nielsen2000}).
However, the matrices of concern in
condition (b) of Theorem~\ref{thm-main1}, $R_i R_j^*$ and $S_i S_j^*$, are general matrices and
may not be Hermitian and positive semi-definite. 
Nevertheless, we prove in Appendix~\ref{appendix1} that
quantum channels are contractive under any unitarily invariant norm 
for general matrices, of which the following theorem for the trace norm is a special case.
\begin{theorem}
\label{thm-trace-norm-contractive}
Given a TPCP map (quantum channel) 
${\mathcal F}: M_q \rightarrow M_p$ 
described by Kraus operators 
$F_1, \dots, F_r \in M_{p, q}$ 
with $\sum_{j=1}^r F_j^*F_j = I_q$
acting on matrices $\sigma \in M_q$ (not necessarily quantum states),
$\trn | {\mathcal F} (\sigma) | \leq  \trn | \sigma |$.
\end{theorem}

\begin{corollary}
If
$$
d(R_i R_j^*,R_{i'} R_{j'}^*) <
d(S_i S_j^*,S_{i'} S_{j'}^*)
$$
for some $i,j,i',j'$, then
condition (b) of Theorem~\ref{thm-main1} does not hold.
\end{corollary}
We prove by contradiction.
If condition (b) holds, then
there exists some quantum channel $\mathcal F$ satisfying the transformations,
and
\begin{eqnarray*}
d(S_i S_j^*,S_{i'} S_{j'}^*)
&=&
\frac{1}{2} \tr |
{\mathcal F}(R_i R_j^*)-
{\mathcal F}(R_{i'} R_{j'}^*)
|
\\
&=&
\frac{1}{2} \tr |
{\mathcal F}(R_i R_j^*-
R_{i'} R_{j'}^*)
|
\\
&\leq&
\frac{1}{2} \tr |
R_i R_j^*-
R_{i'} R_{j'}^*
|
\\
&=&
d(R_i R_j^*,R_{i'} R_{j'}^*)
\end{eqnarray*}
where the inequality is due to Theorem~\ref{thm-trace-norm-contractive}.
\qed

Therefore, if we find the distance between the inputs of two transformations to be smaller than the distance between the outputs, the state is not degradable in the sense of Theorem~\ref{thm-main1}.

\begin{example}
{\rm
Consider the output state processed by the qubit depolarizing channel:
\begin{equation}
{\mathcal E}(\rho) = (1-\epsilon) \rho + \frac{\epsilon}{3} Z \rho Z
+ \frac{\epsilon}{3} Y \rho Y
+ \frac{\epsilon}{3} X \rho X
\label{eqn-depolarizing-channel}
\end{equation}
where $\rho \in M_2$ is the input density matrix,
$X=
\pmatrix{
0&1 \cr
1&0}
$,
$Y=
\pmatrix{
0&-1 \cr
1&0}
$, and
$Z=
\pmatrix{
1&0 \cr
0&-1}
$.
Suppose the input state is 
{\textcolor{mycolor}{
$(\ket{00}+\ket{11})_{AB}$ 
}}%
and $\mathcal E$ is applied to system $B$.
The output state purified with system $E$ is
{\textcolor{mycolor}{
\begin{eqnarray}
\ket{\Psi}_{ABE}=
\sqrt{1-\epsilon} \:\: (\ket{00}+\ket{11})_{AB} \ket{0}_E &+&
\sqrt{\frac{\epsilon}{3}} \:\: (\ket{00}-\ket{11})_{AB} \ket{1}_E +
\nonumber \\
\sqrt{\frac{\epsilon}{3}} \:\: (\ket{01}-\ket{10})_{AB} \ket{2}_E &+&
\sqrt{\frac{\epsilon}{3}} \:\: (\ket{01}+\ket{10})_{AB} \ket{3}_E .
\end{eqnarray}
Note that this state is unnormalized, and normalization is not important in the following discussion.}}%
Denote the coefficient for $\ket{ijk}_{ABE}$ by $x_{ijk}$.
Then, following Theorem~\ref{thm-main1},
\begin{eqnarray}
S_0&=&\pmatrix{ 
x_{000} & x_{001} & x_{002} & x_{003} \cr
x_{010} & x_{011} & x_{012} & x_{013} }
=
\pmatrix{ 
\alpha & \beta & 0 & 0 \cr
0 & 0 & \beta & \beta
}
\\
S_1&=&\pmatrix{ 
x_{100} & x_{101} & x_{102} & x_{103} \cr
x_{110} & x_{111} & x_{112} & x_{113} }
=
\pmatrix{ 
0 & 0 & -\beta & \beta \cr
\alpha & -\beta & 0 & 0
}
\end{eqnarray}
where $\alpha=\sqrt{1-\epsilon}$ and $\beta=\sqrt{\frac{\epsilon}{3}}$, and
\begin{eqnarray}
R_i = S_i^t .
\end{eqnarray}
We compute the trace distances as follows:
\begin{eqnarray}
R_0 R_0^* - R_1 R_1^* &=&
\pmatrix{
0 & 2\alpha \beta & 0 & 0 \cr
2\alpha \beta & 0 & 0 & 0 \cr
0 & 0 & 0 & 2\beta^2 \cr
0 & 0 & 2\beta^2 & 0
}
\\
S_0 S_0^* - S_1 S_1^* &=&
\pmatrix{
\alpha^2 - \beta^2 & 0 \cr
0 & -(\alpha^2 - \beta^2)
}
\end{eqnarray}
\begin{eqnarray}
R_0 R_1^* - R_1 R_0^* &=&
\pmatrix{
0 & 0 & -2\alpha \beta & 0 \cr
0 & 0 & 0 & 2\beta^2 \cr
2\alpha \beta & 0 & 0 & 0 \cr
0 & -2\beta^2 & 0 & 0
}
\\
S_0 S_1^* - S_1 S_0^* &=&
\pmatrix{
0 & \alpha^2 - \beta^2 \cr
-(\alpha^2 - \beta^2) & 0
} .
\end{eqnarray}
It can be shown that
$R_0 R_0^* - R_1 R_1^*$ and
$R_0 R_1^* - R_1 R_0^*$
have singular values
$2 \alpha \beta, 2 \alpha \beta, 2\beta^2, 2\beta^2$,
and
$S_0 S_0^* - S_1 S_1^*$ and
$S_0 S_1^* - S_1 S_0^*$
have singular values
$(\alpha+\beta)(\alpha-\beta), (\alpha+\beta)(\alpha-\beta)$.
Here, we assume $\alpha>\beta$.
Thus, 
$d_R \equiv d(R_0 R_0^* , R_1 R_1^*)=d(R_0 R_1^* , R_1 R_0^*)=2\beta(\alpha+\beta)$
and
$d_S \equiv d(S_0 S_0^* , S_1 S_1^*)=d(S_0 S_1^* , S_1 S_0^*)=(\alpha+\beta)(\alpha-\beta)$.

Therefore, 
we have the condition for the input distance being smaller than the output distance:
\begin{equation}
d_R < d_S
\Rightarrow
\epsilon < \frac{1}{4}.
\label{eqn-example-depolarizing-channel-epsilon}
\end{equation}
Under this condition, there does not exist a quantum channel $T_E$ acting on system $E$ such that $T_E(\rho_{AE})=\rho_{AB}$.
Here, $\rho_{AE}=\tr_B (\ket{\Psi}_{ABE}\bra{\Psi})$ and
$\rho_{AB}=\tr_E (\ket{\Psi}_{ABE}\bra{\Psi})$.
%{\textcolor{mycolor}{
By Theorem~\ref{thm-connection-3},
the same conclusion holds for all other Bell states serving as the input state since all Bell states are unitarily transformable to each other.
We can interpret the result in the context of quantum key distribution (QKD)~\cite{Bennett1984,Ekert1991}, in which
two legitimate parties, conventionally named Alice and Bob (they correspond to systems $A$ and $B$ here), want to share a secret key against an eavesdropper Eve (system $E$ here), by exchanging quantum states.
These states may be modified by Eve.
In a typical QKD session, Alice and Bob learn about the quantum states 
by comparing measurement results in various measurement bases (such as $X$, $Y$, or $Z$).
%}}
%{\textcolor{mycolor2}{
For each basis, we can compute the fraction of measurement mismatches,
%This comparison is usually summarized in a quantity 
which is known as the quantum bit error rate (QBER).
Note that the QKD protocol described here operates in a two-dimensional space, although the presentation of this paper treats arbitrary finite dimensions.
%}}
%{\textcolor{mycolor}{
Since the state in Eq.~(\ref{eqn-depolarizing-channel})  is symmetric with respect to  measurements in $X$, $Y$, and $Z$, the QBER for each of them is the same, $2 \epsilon /3$.
(This means that measurements in say the $X$ basis produce an error rate of $2 \epsilon /3$ when the channel input is an $X$ eigenstate.)
Combining with Eq.~\eqref{eqn-example-depolarizing-channel-epsilon}, 
it means that when the QBER is less than $1/6$, Eve is not able to imitate Bob.
Recall that if, on the other hand, Eve is able to imitate Bob, 
no key can be generated using one-way postprocessing~\cite{Bruss:1998:cloning,Nowakowski:2009:testforcapacity,Moroder:2006:upperbound}.
Thus, our result here is consistent with the result that positive key rate is achievable
when the QBER is less than $1/6$
for the six-state protocol~\cite{Renner:2005:securityproof}.
%}} % color
}
\end{example}

\section{Additional remarks and questions}
\label{sec-additional-problems}

Direct application of Theorem~\ref{thm-main2}
yields the following.

\begin{proposition}
Use the notation in Theorem~\ref{thm-main1}.
The following are equivalent.

{\rm (a)} There are TPCP maps $T = I_n\otimes T_1$ and $L = I_n \otimes L_1$
such that $T(X_2) = X_3$ and $L(X_3) = X_2$.

{\rm (b)}
There is a TPCP map sending
$\{R_uR_v^*: 1 \le u, v \le n\}$ to  
$\{S_uS_v^*: 1 \le u, v \le n\}$, and a TPCP map sending   
$\{S_uS_v^*: 1 \le u, v \le n\}$ 
to $\{R_uR_v^*: 1 \le u, v \le n\}$.

{\rm (c)} There are $p\times q$ matrices $F_1, \dots, F_r$ with $\sum_{j=1}^r F_j^*F_j = I_q$
and $k_i \times p$ matrices $W_{i1}, \dots, W_{ir}$ 
such that for all $i, j = 1, \dots, n$,
$$[F_1R_i \cdots F_r R_i] = V_iD_i[W_{i1} \cdots W_{ir}]
\quad \hbox{ and } \quad [W_{i1} \cdots W_{ir}][W_{j1} \cdots W_{jr}]^* = U_i^t\overline{U}_j,$$
and
there are $q\times p$ matrices $\tilde F_1, \dots, \tilde F_s$ with $\sum_{j=1}^s \tilde F_j^*
\tilde F_j = I_p$
and $k_i \times q$ matrices $\tilde W_{i1}, \dots, \tilde W_{is}$ 
such that for all $i, j = 1, \dots, n$,
$$[\tilde F_1R_i^t \cdots \tilde F_s R_i^t] = U_iD_i[\tilde W_{i1} \cdots \tilde W_{is}]
\quad \hbox{ and } \quad [\tilde W_{i1} \cdots \tilde W_{is}][\tilde W_{j1} \cdots \tilde W_{js}]^* 
= V_i^t\overline{V}_j.$$
\end{proposition}
\begin{proposition} The following are equivalent.

{\rm (a)} 
There is a TPCP map sending $R_iR_j^*$ to $S_iS_j^*$.

{\rm (b)} There is a TPCP map sending 
 $(\sum c_i R_i)(\sum \tilde c_j R_j)^*$ to $(\sum c_i S_i)(\sum \tilde c_j S_j)^*$
 for any scalars $c_1, \dots, c_n, \tilde c_1, \dots, \tilde c_n$.
 
{\rm (c)} There is a TPCP map sending 
 $(\sum c_i R_i)(\sum  c_j R_j)^*$ to $(\sum c_i S_i)(\sum c_j S_j)^*$
 for any scalars $c_1, \dots, c_n$.
 
\end{proposition}

By the above proposition, we can focus on a maximal linearly independent subset set 
$\{R_1, \dots, R_m\}$ and check whether there is a TPCP map sending 
$R_iR_j^*$ to $S_iS_j^*$ for matrices $R_i, R_j$ in this set.
Note, however, that the above propositions are not very practical and it is desirable to have some more practical conditions.

\medskip\noindent
{\bf Problem}
Can we extend Corollary~\ref{cor-main3}
and determine the condition for the existence of 
TPCP maps $T = I_n \otimes T_1$ and $L = I_n\otimes L_1$ 
such that $T(X_2) = X_3$ and $L(X_3) = X_2$?

\section{Concluding remarks}
\label{sec-conclusion}

{\textcolor{mycolor2}{
In this paper, 
we introduced the notion of state degradability.
The joint state of Alice and Eve is degradable if Eve's system can be processed by a quantum channel to produce a joint state that is the same as the joint state of Alice and Bob.
We proved necessary and sufficient conditions for 
state degradability.
}}%
%the problem of determining whether a given state is degradable by processing one subsystem via a quantum channel.
The conditions are in general difficult to check, but we also provide an easily computable method to rule out degradability.
This method is based on the fact that the trace distance between two states can only become smaller under the action of a quantum channel.
{\textcolor{mycolor2}{
One application of state degradability is that it can be used to test channel degradability.
Analysis of the channel output state of the maximally entangled input state gives information about the degradability of the channel.
}}%
Another application of state degradability is in the analysis of QKD, in which no secret key can be generated by one-way postprocessing when the joint state between Alice and Eve can be degraded to a joint state between Alice and Bob.
For future work, we hope to investigate more connections between degradability and other quantum information processing tasks, and extend our result to the case where Alice, Bob, and Eve share a mixed quantum state.

\appendix

\section{Proof of Theorem~\ref{thm-trace-norm-contractive}}
\label{appendix1}

A norm is unitarily invariant if $\|X\| = \|UXV\|$ for any unitary $U,V$.
Note that the trace norm is one such norm.

\begin{proposition}  Suppose $B \in M_p$ and $A\in M_q$ such that
$B = \sum_{j=1}^r F_j A F_j^*$ with $\sum F_j^*F_j = I_q.$
Then,
$\|I_r \otimes B\| \le r \|A \oplus O\|$ for any 
   unitarily invariant norm $\|\cdot\|$ on $M_{pr}$.
\end{proposition}

\it Proof. \rm   Let $U = (U_{ij})_{1 \le i,j \le r}$ 
be unitary such that $U_{j1} = F_j$ for $j = 1, \dots, r$.
Then $U(A \oplus O)U^* = (A_{ij})_{1 \le i,j \le r}$ such that 
    $A_{11} + ... + A_{rr} = B$.
   Take $P = \diag(1,w, .., w^{r-1}) \otimes I_p$ with $w = e^{i2\pi/r}$.
   Then  ${r}^{-1} \sum_{1 \le \ell \le r} P^\ell(A_{ij})(P^\ell)^* 
   = A_{11} \oplus \cdots \oplus A_{rr}$.
Now take $Q = (E_{12} + \cdots E_{r-1,r} + E_{r,1})\otimes I_p$.
Then  $\sum_{1 \le j \le r} Q^j(A_{11} 
\oplus \cdots \oplus A_{rr})(Q^j)^* = I_r \otimes B$.
Thus, using the triangle inequality,
$$\| I_r \otimes B\| =
\left\| \frac{1}{r}\sum_{\ell,k=1}^r Q^k P^\ell U (A\oplus O)(Q^k P^\ell U)^* \right\|
\le
\frac{1}{r}\sum_{\ell,k=1}^r \| A\oplus O \|
=r\| A\oplus O \| .
$$
\qed

To prove Theorem~\ref{thm-trace-norm-contractive}, we just 
take $\|\cdot\|$ to be the trace norm to get
$$
\| B \| = \frac{1}{r} \|I_r \otimes B\|  \le  \| A\oplus O \|=  \| A \|.
$$

\section*{Acknowledgments}%
We thank Zejun Huang and Edward Poon for enlightening discussion.

This research evolved in a faculty seminar on quantum information science
at the University of Hong Kong in the spring of 2012 coordinated by
Chau and Li. The support of the Departments of Physics and Mathematics 
of the University of Hong Kong is greatly appreciated.

Chau and Fung were partially supported by the Hong Kong RGC grant
No. 700712P. 
Sze was partially supported by the Hong Kong RGC grant PolyU 502512.
Li
was supported
by a USA NSF grant and a Hong Kong RGC grant; 
he was a visiting professor of the  University of Hong Kong in the
spring of 2012,  an honorary professor of Taiyuan University of Technology
(100 Talent Program scholar), and an honorary  professor of  Shanghai
University.

\bibliographystyle{unsrt}

\bibliography{paperdb}

\end{document}